\documentclass[graybox,envcountsect]{svmult}

\usepackage{helvet}         % selects Helvetica as sans-serif font
\usepackage{courier}        % selects Courier as typewriter font
\usepackage{type1cm}        % activate if the above 3 fonts are
\usepackage{makeidx}         % allows index generation
\usepackage{graphicx}        % standard LaTeX graphics tool
\usepackage{multicol}        % used for the two-column index
\usepackage[bottom]{footmisc}% places footnotes at page bottom

\usepackage{amsmath}
\usepackage{amssymb}
\usepackage{eucal}
\usepackage{mathrsfs}
\usepackage{theorem}
\usepackage{amscd}
\usepackage{color}
\usepackage{mathbbol}
\usepackage{enumitem}

\usepackage{xcolor} % Required for specifying colors by name
\definecolor{ocre}{RGB}{243,102,25} 
\definecolor{darkocre}{RGB}{121,51,12} 
\definecolor{lightocre}{RGB}{255,150,37} 
\definecolor{verylightocre}{RGB}{255,204,50} 
\definecolor{soccerfield}{RGB}{107,142,35} 
\definecolor{lightgray}{RGB}{200,200,200} 
\definecolor{warmblue}{RGB}{51,102,153} 
\definecolor{lightwarmblue}{RGB}{105,141,198} 
\definecolor{sepia}{RGB}{112,66,20}

\usepackage{amsmath}
\usepackage{amssymb}
\usepackage{graphicx}
\usepackage{eucal}
\usepackage{mathrsfs}
\usepackage{pifont}
\usepackage{color}
\usepackage{cjhebrew}

\usepackage[all]{xy}
\usepackage{tikz} % Required for drawing custom shapes
\usepackage{tkz-euclide}
\usetkzobj{all}
\usepackage{pgfplots}

\newcommand{\btkz}{\begin{tikzpicture}}
\newcommand{\etkz}{\end{tikzpicture}}

\newcommand{\brk}[1]{\left(#1\right)}          % \brk{.}     => (.)
\newcommand{\Brk}[1]{\left[#1\right]}          % \Brk{.}     => [.]
\newcommand{\Abs}[1]{\left| #1 \right|}        % \Abs{.}     => |.|
\newcommand{\Norm}[1]{\left\| #1 \right\|}     % \Norm{.}    => ||.||

\newcommand{\Cases}[1]{\begin{cases} #1 \end{cases}}

\newcommand{\pd}[2]{\frac{\partial#1}{\partial#2}}
\newcommand{\pdd}[2]{\frac{\partial^2#1}{\partial#2^2}}

\newcommand{\secref}[1]{Section~\ref{#1}}
\newcommand{\figref}[1]{Figure~\ref{#1}}

\newcommand{\propref}[1]{Proposition~\ref{#1}}

\newcommand{\beq}{\begin{equation}}
\newcommand{\eeq}{\end{equation}}
\newcommand{\bsplit}{\begin{split}}
\newcommand{\esplit}{\end{split}}
\newcommand{\baligned}{\begin{aligned}}
\newcommand{\ealigned}{\end{aligned}}

\newcommand{\Emph}[1]{{\slshape\bfseries #1}}  % \Emph{.}

\providecommand{\e}{\varepsilon}
\providecommand{\half}{\frac{1}{2}}

\providecommand{\R}{\bbR}

\newcommand{\Textand}{\qquad\text{ and }\qquad}

\newcommand{\Hom}{{\operatorname{Hom}}}

\newcommand{\limn}{\lim_{n\to\infty}}

\newcommand{\calL}{{\mathcal L}}
\newcommand{\calM}{{\mathcal M}}

\newcommand{\calT}{{\mathcal T}}

\newcommand{\frakg}{\mathfrak{g}}

\newcommand{\bbN}{{\mathbb N}}

\newcommand{\bbR}{{\mathbb R}}

\newcommand{\scrD}{\mathscr{D}}

\newcommand{\M}{\calM}
\newcommand{\cof}[1]{\vartheta^{#1}}
\newcommand{\g}{\frakg}
\renewcommand{\Emph}[1]{{\bfseries #1}}

\newcommand{\Mleft}{\M_{\text{left}}}
\newcommand{\Mright}{\M_{\text{right}}}
\newcommand{\Mtop}{\M_{\text{top}}}
\newcommand{\Mbottom}{\M_{\text{bottom}}}

\DeclareMathOperator{\supp}{supp}

\makeindex             % used for the subject index
\begin{document}

\title*{Homogenization of edge-dislocations as a weak limit of de-Rham currents}
% Use \titlerunning{Short Title} for an abbreviated version of
% your contribution title if the original one is too long
\author{Raz Kupferman and Elihu Olami}
% Use \authorrunning{Short Title} for an abbreviated version of
% your contribution title if the original one is too long
\institute{Raz Kupferman \at Institute of Mathematics, The Hebrew University, \email{raz@math.huji.ac.il}
\and Elihu Olami \at Institute of Mathematics, The Hebrew University \email{elikolami@gmail.com}}
%
% Use the package "url.sty" to avoid
% problems with special characters
% used in your e-mail or web address
%
\maketitle

%\abstract*{Each chapter should be preceded by an abstract (10--15 lines long) that summarizes the content. The abstract will appear \textit{online} at \url{www.SpringerLink.com} and be available with unrestricted access. This allows unregistered users to read the abstract as a teaser for the complete chapter. As a general rule the abstracts will not appear in the printed version of your book unless it is the style of your particular book or that of the series to which your book belongs.
%Please use the 'starred' version of the new Springer \texttt{abstract} command for typesetting the text of the online abstracts (cf. source file of this chapter template \texttt{abstract}) and include them with the source files of your manuscript. Use the plain \texttt{abstract} command if the abstract is also to appear in the printed version of the book.}

%\abstract{Each chapter should be preceded by an abstract (10--15 lines long) that summarizes the content. }

\abstract{In the material science literature we find two continuum models for crystalline defects: (i) A body with (finite) isolated defects is typically modeled as a Riemannian manifold with singularities, and (ii) a body with continuously distributed defects, which is modeled as a smooth (non-singular) Riemannian manifold with an additional structure of an affine connection. In this work we show how continuously distributed defects may be obtained as a limit of singular ones .   The defect structure is represented by layering $1$-forms and their singular counterparts - de-Rham $(n-1)$ currents.  We then show that every smooth layering $1$-form may be obtained as a limit, in the sense of currents, of singular layering forms, corresponding to arrays of edge dislocations. As a corollary,  we investigated manifolds with full material structure, i.e., a complete co-frame for the co-tangent bundle. We define the notion of singular torsion current for manifolds with a parallel structure and prove its convergence to the regular smooth torsion tensor at homogenization limit. Thus establishing the so-called emergence of torsion at the homogenization limit. }
%%%%%%%%%%%%%%%%%%%%%%%%%%%%%%%%%%%%%

%%%%%%%%%%%%%%%%%%%%%%%
\section{Introduction}
\label{sec:intro}

The study of material defects, and notably dislocations, is a central theme in material science. 
The modeling of solid bodies, with or without defects, often follows a paradigm in which the
elemental object is that of a \emph{body manifold}: solid bodies are modeled as geometric objects---manifolds---and their internal structure is represented by additional structure such as a frame field, a metric or an affine connection. The mechanical properties of the body enter through a \emph{constitutive relation}, whose structure is correlated with the geometric structure of the body.

There have been two distinct approaches to the modeling of body manifolds with dislocations:

\begin{enumerate}
\item \Emph{Isolated dislocations}: One starts with a defect-free body, which is either modeled as a compact subset of Euclidean space, or as a perfect lattice. Defects are introduced by Volterra cut-and-weld protocols \cite{Vol07}; see \figref{fig:1}. Note that a perfect lattice may be related to a Euclidean structure by assigning lengths and angles to inter-particle bonds.  

\begin{figure}
\begin{center}
\includegraphics[height=2in]{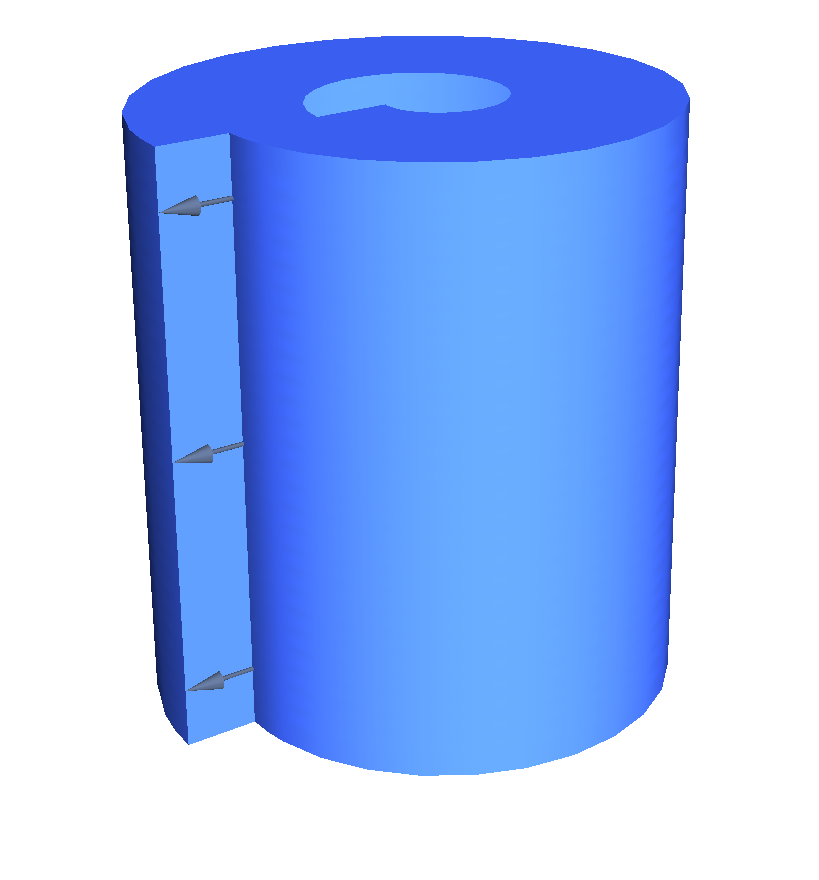}
\includegraphics[height=1.8in]{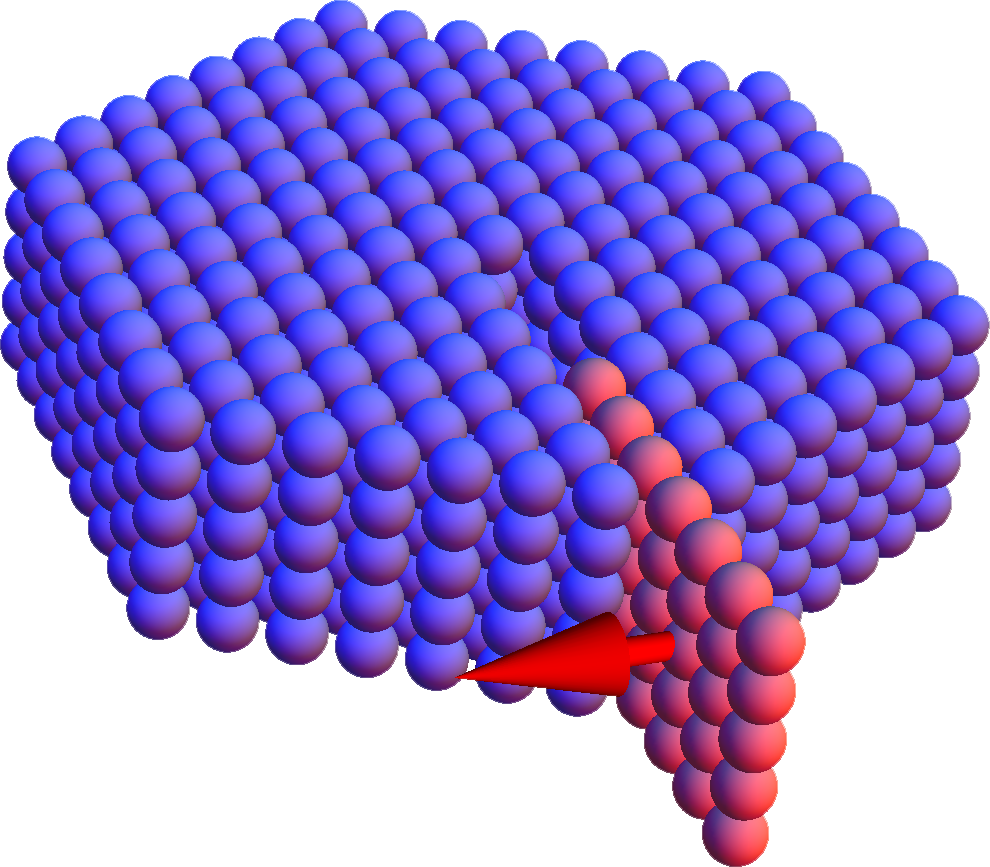}
\end{center}
\caption{ Left: An edge-dislocation generated by a cut-and-weld protocol in a continuum setting. Right: An edge-dislocation generated by removing a half-plane in a lattice.}
\label{fig:1}
\end{figure}

\item \Emph{Distributed dislocations}: In the classical literature from the 1950s, 
the body is modeled as a smooth manifold endowed with a curvature-free affine connection \cite{Nye53,Kon55,BBS55,Kro81}. 
If, in addition, one adds a basis of the tangent space at one point, then the affine connection induces a smooth frame field, which is the kinematic model, for example, in \cite{Dav86}.
In later literature \cite{YG12}, the continuum model is that of a Weitzenb\"ock manifold, which is a smooth manifold endowed with a Riemannian metric and a metrically-consistent, curvature-free affine connection (in fact, the vanishing curvature condition has to be replaced by the even stronger condition of trivial holonomy). Note that a frame field induces an intrinsic metric, so that all three descriptions are essentially identical.
The density of the dislocations is identified with the torsion tensor of the affine connection.
\end{enumerate}

A longstanding problem has been to rigorously justify the continuum model of distributed dislocations as a limit of  (properly scaled) isolated dislocations, as their number tends to infinity, in the spirit of other homogenization theories. Such an analysis was recently presented in \cite{KM15,KM16}. Specifically, bodies with either isolated or distributed dislocations were modeled as Weitzenb\"ock manifolds $(\M,\g,\nabla)$. In the case of isolated dislocations, the smooth part of the manifold is multiply-connected, the defects being located inside either non-smooth sets, or ``holes", and the connection is the Riemannian (Levi-Civita) connection. A sequence of multiply-connected body manifolds with isolated dislocations may converge to a simply-connected body manifold $(\M,\g,\nabla)$, where $\nabla$ is non-symmetric; that it, torsion arises as a weak limit of torsion-free connections. Moreover, it was shown that every triple $(\M,\g,\nabla)$ can be obtained as a limit of bodies with isolated dislocations. 

The work \cite{KM15,KM16} has several shortcomings: (i) The notion of convergence was tailored to the problem, and as a result, does not coincide with prevalent notions of convergence. 
%(ii) It assumes more structure than needed. The homogenization limit as well as the emergence of torsion should be obtained regardless of the metric structure.  
(ii) The standard mathematical apparatus accounting for singularities is generalized functions (or generalized sections), thus providing a natural setting for convergence. This has been missing here. (iii) In particular, one would hope to recover a notion of singular torsion, in the same spirit as one obtains a notion of singular curvature for cone singularities. (iv) This analysis requires the consideration of a complete lattice structure, not including, for example, scalar elastic invariants.

An alternative approach to defects, and notably to dislocations, was proposed by Epstein and Segev \cite{ES14}. Their point of view is that every material structure is represented by one or more differential forms: while smooth structures are represented by smooth differential forms, singularities in structure are represented by their distributional counterparts---\emph{de-Rham currents}.

Specifically, \cite{ES14} models the structure of a lattice by means of differential forms termed \emph{layering forms} which represent \emph{Bravais surfaces}. In $n$ dimensions, the prescription of a set of $n$ linearly-independent 1-forms  $\cof{i}$ (a coframe) amounts to Davini's frame field approach \cite{Dav86}, but the points in \cite{ES14} are: 
\begin{enumerate}
\item A single layering form may suffice to track the presence of defects.
\item Layering forms can be singular. 
\end{enumerate}

In the absence of defects, the layering forms are closed, namely,
\[
d\cof{i} = 0.
\]
In the case of distributed defects, we expect $d\cof{i}\ne0$, where the 2-form $d\cof{i}$ is related to the density of the defects.

The prescription of $n$ linearly-independent layering forms defines an \emph{intrinsic metric}, $\sum_i \cof{i}\otimes \cof{i}$ and a \emph{material connection}, whose path-independent parallel transport $\Pi_p^q:T_p\M\to T_q\M$ between two points $p,q\in\M$ is given by
\[
\Pi_p^q = {e_i}|_q \otimes {\cof{i}}|_p,
\] 
where $\{e_i\}$ is the frame field dual to $\{\cof{i}\}$ (here and below we adopt Einstein's summation convention).
The generalization of this approach to structures with singularities is as follows: to every smooth 1-form $\cof{i}$ corresponds an $(n-1)$-current
\[
T_{\cof{i}}(\alpha) = \int_\M \cof{i}\wedge\alpha,
\qquad
\alpha\in\Omega^{n-1}_c(\M),
\]
where $\Omega^{k}_c(\M)$ denotes the module of smooth, compactly-supported $k$-forms on $\M$ (see \secref{sec:currents} for a short review of de-Rham currents).
By definition, the boundary of an $(n-1)$-current is an $(n-2)$-current
\[
\partial T_{\cof{i}}(\beta) = \int_\M \cof{i}\wedge d\beta=\int_\M d\cof{i}\wedge \beta
\qquad
\beta\in\Omega^{n-2}_c(\M),
\]
where the second identity follows from integration by parts and the compact support of $\beta$.
If $\cof{i}$ is closed then $\partial T_{\cof{i}}=0$, i.e., the absence of defects is reflected by the vanishing of the boundary of the current induced by $\cof{i}$. Just like in classical distribution theory, not every $(n-1)$-current is induced by a smooth 1-form; 
 structures with singularities are modeled by $(n-1)$-currents; the defects are associated with the boundary of those currents that are not induced by smooth forms.

In this work, we show that the homogenization of singular defects can be cast in the framework of weak convergence of currents. To set the stage, we review 
in \secref{sec:currents} some basic facts about de-Rham currents. % and other geometric constructions used below. 
In \secref{sec:example}, we consider an arbitrary (generally non-closed) smooth layering form $\beta\in \Omega^1(\M)$ on the  two-dimensional square $\M=[0,1]^2$, which we view as representing  distributed edge-dislocations. We develop a generic construction of a  layering form $\nu$, which approximate $\beta$ (in a sense made precise), while being smooth and closed everywhere, except on a one-dimensional sub-manifold $\Gamma$. Furthermore, interpreting the layering form $\nu$ as a 1-current, we show that the boundary of that current is supported on $\Gamma$. 
Thus, view the layering form $\nu$ as representing a singular edge-dislocation, whose locus is $\Gamma$, and whose intensity is equal to the total intensity of the layering form $\beta$.

%construct a layering form $\nu$ for the two-dimensional square $\M=[0,1]^2$, corresponding to a single edge-dislocation; specifically,  $\nu$ is smooth and closed everywhere, except for a one-dimensional sub-manifold $\Gamma$---the locus of the dislocation. \Elik{The boundary values of $\nu$ as well as its circulation on the boundary $\partial\M$ are dictated by an arbitrary  non-closed $1$-form $\beta\in\Omega^1(\M)$.   Expressing the layering form $\nu$ as} a 1-current, we observe that the boundary of that current, which is a 0-current, is supported on $\Gamma$. 

In \secref{sec:homogenization},we show that every (possibly non-closed) $1$-form $\beta\in\Omega^1(\M)$ can be approximated by a sequence of discontinuous layering forms $\nu^{(n)}$, representing an $n$-by-$n$ array of edge-dislocations. We construct $\nu^{(n)}$ by gluing together properly rescaled versions of the form constructed in \secref{sec:example}. We then prove that $T_{\nu^{(n)}}$ converges as $n\to\infty$ to a 1-current $T_\beta$; the convergence is in the sense of weak convergence of currents. 
We interpret this limit theorem as a statement that every smooth distribution if dislocations is a limit, in the sense of weak convergence of currents, of singular dislocations.
%We interpret $\beta$ as a layering form of a body with continuously-distributed dislocations. In particular, the boundary of current $T_\beta$ is supported on a set of full measure.

%In \secref{sec:homogenization}, \Elik{given a smooth $1$-form $\beta\in\Omega^1(\M)$, corresponding to a smooth distribution of dislocations, can be approximated in the sense of currents by  singular  layering $!$-forms  corresponding to arrays of isolated dislocations.} We glue together properly rescaled versions of $\nu$, thus forming an $n$-by-$n$ array of edge-dislocation; we denote this discontinuous layering form by $\nu^{(n)}$, and the corresponding 1-current by $T_{\nu^{(n)}}$. We then prove that $T_{\nu^{(n)}}$ converges as $n\to\infty$ to a 1-current $T_\beta$, induced by a smooth non-closed 1-form $\beta$; the convergence is in the sense of weak convergence of currents. We interpret $\beta$ as a layering form of a body with continuously-distributed dislocations. In particular, the boundary of $T_\beta$ is supported on a set of full measure. 

In \secref{sec:torsion}, we generalize the analysis to the case where $\M$ is an $n$ dimensional manifold equipped with a full lattice structure, that is, a (possibly singular) frame field $\{e_i\}_{i=1}^n$. We cast in the setting of currents the convergence of parallel transport and torsion. In particular, we define the notion of singular torsion, and show that ``the emergence of torsion" as a limit of torsion-free connections, as exposed in \cite{KM15}, should be re-interpreted as a convergence of singular torsions to a limiting smooth torsion. 
%In \secref{sec:generic}, we build upon a theorem in \cite{} to show that the homogenization procedure presented in \secref{sec:homogenization} is generic: every non-closed layering form is a homogenization limit of TO COMPLETE.
Further extensions and concluding remarks are presented in \secref{sec:discussion}.

%%%%%%%%%%%%%%%
%\section{Mathematical preliminaries}
%\label{sec:prelim}
% 
  %%%%%%%%%%%%%%
\section{De-Rham currents}
\label{sec:currents}

We start by reviewing the definition of de-Rham currents on manifolds, which are  fundamental objects representing singular material structures. 
%\Elik{De-Rham} currents were first introduced  by de-Rham \cite{} and in a different context by Whitney \cite{}. 
For a full introduction see for example the classical monograph of Federer \cite{Fed69} or de-Rhams  \cite{deR84}. For more recent reviews see also \cite{LY02,KP08}. 
 
Let $\M$ be a smooth, compact, orientable $n$-dimensional manifold with boundary. For every $1\leq k\leq n$, let $\Omega^k(\M)$ denote the space of smooth $k$-forms on $\M$ and let
\[
\Omega_c^k(\M)=\left\{\omega\in \Omega ^k(\M) ~:~  \supp(\omega)\Subset \M  \right\}
\]
be the module of smooth  $k$-forms compactly-supported in $\M$. Choose a Riemannian metric $g$ on $\M$, and define for every compact $K\Subset \M$ a family of seminorms $\phi_{K,j}^k: \Omega_c^k(\M)\to\R^+$ by 
\[
\phi_{K,j}^k(\omega)=\sup_{0\leq i\leq j} \|D^i\omega\|_{K},                                                                                                                                                        
\] 
where $D^i\omega: \M \to \Hom(\otimes^i T\M,\Lambda^kT^*\M)$ is the $i$-th differential of $\omega$ (not to be confused with the exterior derivative), and 
\[
\|D^i \omega\|_{K}=\sup_{p\in K} \| (D^i\omega)_p\|,
\] 
where $\|\cdot\|$ is the norm  on $\Hom(\otimes^i T\M,\Lambda^k T^*\M)$ induced by the metric $g$. Since $\M$ is compact, a different choice of $g$ will give equivalent seminorms; as a result, it makes sense to say that a $k$-form is $C^j$-bounded without reference to any metric. The seminorms $\phi_{K,j}^k$ turn 
\[
\Omega_K^k(\M)=\{ \omega\in \Omega_c^k(\M) ~:~ \supp(\omega)\subset K  \}
\] 
into a Fr\'echet space, that is, a locally-convex topological vector space which is complete with respect to a translationally-invariant metric \cite[p. 9]{Rud91}. 

Endow $\Omega_c^k(\M)$ with the finest topology for which the inclusion maps 
\[
\Omega_K^k(\M)\hookrightarrow \Omega_c^k(\M)
\]
are continuous for all compact $K\Subset \M$.  It follows that a sequence $\omega_n\in \Omega_c^k(\M)$ converges in this topology to $0$ if and only if there exists a compact set $K\Subset \M$ such that $\supp(\omega_n)\subset K$ for all $n$, and $\omega_n\to 0$ in the topology of $\Omega_K^k(\M)$ described above. 

Finally, let $\scrD_k(\M)$ be the dual vector space of  continuous linear functionals on $\Omega_c^k(\M)$; the members of $\scrD_k(\M)$ are called \emph{de-Rham $k$-currents}. 
Equivalently, a linear functional $T:\Omega_c^k(\M)\to \R$ is a $k$-current if and only if there exists for every $K\Subset\M$ an $N=N(K)\in\mathbb{N}$ and a constant $C=C(K)>0$, such that for every $\omega\in \Omega_K^k(\M)$, 
\[
|T(\omega)|\leq C\, \sup_{1\leq j\leq N}\phi_{K,j}^k(\omega).
\]
We endow $\scrD_k(\M)$ with the weak-star topology: a sequence of $k$-currents $T_n$ converges to a $k$-current $T$ if 
\[
\lim_{n\to \infty}T_n(\omega)=T(\omega)
\]
for every $\omega \in \Omega_c^k(\M)$.
The \emph{support} of a $k$-current $T\in \scrD_k(\M)$ is defined by $\supp(T)=\M\setminus A(T)$, where $A(T)$ is the annihilation set of $T$, i.e., the union of all open subsets $U\subset\M$ for which $T(\alpha)=0$ whenever $\supp(\alpha)\subset U$.

For example, every locally-integrable $k$-form $\beta$ defines an $(n-k)$-current $T_\beta\in \scrD_{n-k}(\M)$ by 
\[
T_\beta(\alpha)=\int_\M \beta\wedge \alpha.
\] 
In other words, currents may be viewed as generalized differential forms. Currents also generalize the concept of a submanifold. Let $S\subset \M$ be a $k$-dimensional oriented submanifold, then $S$ induces a $k$-current $[S]$ given by 
\[
[S](\alpha)=\int_S\alpha,\quad \alpha\in \Omega_c^k(\M).
\]

The \emph{boundary operator} of a $k$-current is a map $\partial: \scrD_k(\M)\to \scrD_{k-1}(\M)$, defined by 
\[
\partial T(\alpha)=T(d\alpha),\quad \alpha\in \Omega_c^{k-1}(\M).
\]
Since $d^2=0$, it immediately follows that $\partial^2=0$;
moreover, 
it follows from integration by parts and Stokes theorem that 
\[
\partial T_{\beta}=(-1)^{k-1} T_{d\beta}
\] 
for every smooth $k$-form $\beta$, 

%%%%%%%%%%%%%%%%%%%%

%%%%%%%%%%%%%%%%%%%%%%%%%%%
\section{Layering form for an edge-dislocation}
\label{sec:example}

As discussed in Epstein \cite[Section~4.5.3]{Eps10} and \cite{SE14},  a single differential 1-form is capable of capturing the presence of a dislocation. A covector $\omega$ in a vector space $V$ induces a family of hyperplanes (Bravais planes), 
\[
H_t=\{v\in V ~:~ \omega(v)=t\}
\] 
foliating $V$; the action of $\omega$ on a vector $v\in V$ can be viewed as ``the number of hyperplanes" intersected by the vector $v$. In the case of a smooth manifold $\M$,  given a  1-form $\nu$ and an oriented curve $C\subset \M$, the integral 
\[
\int_{C} \nu
\]
can be interpreted as the (signed) number of  $\nu$-hyperplanes intersected by $C$. Thus, a single 1-form $\nu$ on a manifold $\M$, can be viewed as representing a \emph{layering form}---a density of a family of parallel layers at each point.

A 1-form $\nu$  induces a smooth layering structure (foliation) for $\M$ if it is integrable; that is, if $\M$ can be foliated such that the tangent bundle of each leaf coincides with the kernel of $\nu$. It is well known that a sufficient and necessary condition for $\nu$ to induce a smooth layering structure is that 
\[
d\nu=\alpha\wedge \nu
\]
for some $(n-1)$-form $\alpha$ \cite[Chap.~19]{Lee12}. Note that for a simply-connected two-dimensional manifold, every non-vanishing 1-form induces a smooth layering structure.

If, in addition, the 1-form $\nu$ is closed, $d\nu=0$, then it follows from Stokes' theorem that for every simple, oriented, closed curve $C\subset\M$, the sum of all the hyperplanes intersected by $C$ vanishes,  
\beq
\label{eq:closed_nu}
\int_C \nu=\int_{\Sigma_C}d\nu=0,
\eeq
where $\Sigma_C\subset\M$ is any $2$ dimensional submanifold of $\M$ bounded by $C$.  In other words, there are no ``extra" layers, and the layering structure is defect-free.  Motivated by equation \eqref{eq:closed_nu}, we may interpret $d\nu$ as a defect density.   

Suppose in turn that $\nu$ is a 1-form corresponding to an  isolated dislocation concentrated on a hyper-surface $\Gamma\subset\M$. By \eqref{eq:closed_nu}, $d\nu=0$ on $\M\setminus\Gamma$, and consequently, $\nu$ must be singular at $\Gamma$.  

We next construct an explicit layering form on a two-dimensional manifold, which may represent a singular edge-dislocation in one family of Bravais planes. We first consider a topological rectangle, i.e., a manifold that can be parametrized as follows:
\[
\M=[0,1]^2  = \{(x,y) ~:~ 0\le x,y \le 1\}.
\]
We denote the left, right, top and bottom edges of $\M$ by $\Mleft,$ $\Mright$, $\Mtop$ and $\Mbottom$, respectively. 
The locus of the dislocation is a one-dimensional submanifold, with we take to be the closed parametric segment
\beq
\Gamma_a = [1/2 - a/2,  1/2 + a/2] \times \{1/2\} \subset \M,
%\{ (x+1/2,1/2) ~:~ -a/2\leq x\leq a/2 \} \subset \M,
\label{eq:Gamma}
\eeq
where $0<a<1$ is a parameter, which will be used later in our homogenization procedure.

%
%and let $\beta\in\Omega^1(\M)$ be the smooth, non-closed 1-form 
%\beq
%\beta = \brk{1+ bx}\, dy.
%\label{eq:beta}
%\eeq
%The role of $\beta$ will be seen below, but we already note that
%$\beta$,  foliates $\M$ as a union 
%\[
%\M=\amalg_{y\in [0,1]}L_y,\quad L_y:=\{(x,y) ~:~ x\in [0,1]\}.
%\]
%The line density of the layers is non-constant, increasing linearly with $x$, \Raz{which reflects in the fact that 
%\[
%d\beta= b\, dx \wedge dy\ne0.
%\]} 
%Thus, $\beta$ induces a smooth layering structure containing defects. 
%\Raz{Moreover, the circulation of $\beta$ around the manifold is
%\[
%\int_{\partial\M} \beta = b.
%\]}
%
%\Elik{PERHAPS WE SHOULD CONSIDER THE GENERAL CASE FROM THE START.}
%\Raz{I'll get back to this point on the next round.}

%\Elik{ Denote by $c:=\int_{\partial\M}\beta$ the circulation of $\beta$.}

%%%%%%%%%%
\begin{proposition}
\label{prop:nu}
Let $\beta\in\Omega^1(\M)$ be a nowhere-vanishing 1-form and let  $0<a<1$.
Then, there exists a continuously differentiable 1-form $\nu_a$ on $\M\setminus \Gamma_a$ satisfying the following properties:
\begin{enumerate}[label=(\roman*)]
\item $\nu_a$ is  $C^1$-bounded (see definition in \secref{sec:currents}).
\item $\nu_a$ is closed.
\item $\nu_a$ coincides with $\beta$ on $\Mleft$ and $\Mright$.
\item $\nu_a$ has the same circulation as $\beta$, 
\[
\int_{\partial\M}\nu_a=\int_{\partial\M}\beta.
\]
\item The horizontal components of $\nu_a$ and $\beta$ coincide,
\[
\nu_a(\partial_x)=\beta(\partial_x), 
\]
whenever $|x-1/2|>a/2$.

%\item 
%Gluing conditions: two copies of $\M\setminus\Gamma_a$ can be glued on opposite edges yielding a smooth layering form: specifically, denote by $\partial\M_{\text{left}}$, $\partial\M_{\text{right}}$,  $\partial\M_{\text{top}}$ and $\partial\M_{\text{bottom}}$ the four smooth components of $\partial\M$. The gluing map 
%\[
%h: \partial\M_{\text{top}} \to \partial\M_{\text{bottom}}
%\]
%defined by $h(x,1) = (x,0)$ is a diffeomorphism, which together with the collar neighborhoods
%\[
%\iota_{\text{top}}:[0,\e)\times \partial\M_{\text{top}} \to \M
%\textand
%\iota_{\text{bottom}}:[0,\e)\times \partial\M_{\text{bottom}} \to \M
%\]
%defined by
%\[
%\iota_{\text{top}}(t,x) = (x,1-t)
%\Textand
%\iota_{\text{bottom}}(t,x) = (x,t)
%\]
%turn $\M\amalg_h \M$ into a smooth manifold, such that two replica of $\nu_a$ combine into a smooth 1-form on $\M\amalg_h \M$. 
%
%Likewise, let $\tilde\M$ be a copy of $\M$ endowed with the 1-form
%\[
%\tilde\nu_a = \nu_a + b\, dy.
%\]
%The gluing map 
%\[
%\tilde{h} : \partial\M_{\text{right}} \to \partial\M_{\text{left}}
%\]
%defined by $\tilde{h}(1,y) = (0,y)$ is a diffeomorphism,  which together with the collar neighborhoods
%\[
%\iota_{\text{right}}:[0,\e)\times \partial\M_{\text{right}} \to \M
%\textand
%\iota_{\text{left}}:[0,\e)\times \partial\M_{\text{left}} \to \M
%\]
%defined by
%\[
%\iota_{\text{right}}(t,y) = (1-t,y)
%\Textand
%\iota_{\text{left}}(t,y) = (t,y)
%\]
%turn $\M\amalg_{\tilde{h}} \tilde\M$ into a smooth manifold, such that $\nu_a$ and $\tilde{\nu}_a$ combine into a smooth 1-form on $\M\amalg_{\tilde{h}} \tilde\M$. 
\end{enumerate}
\end{proposition}
%%%%%%%%%

Before proving \propref{prop:nu}, we show in which sense the 1-form $\nu_a$ represents a family of Bravais planes dislocated along the segment $\Gamma_a$. 
Since $\nu_a$ is closed in $\M\setminus\Gamma_a$, it follows that
\[
\oint_C \nu_a = 0
\]
along every contractible loop $C$ in $\M\setminus\Gamma_a$. 

Let $\g$ be a metric on $\M$, and
denote by $\Gamma_a^\e$, $\e>0$, a family of $\e$-tubular neighborhoods of $\Gamma_a$. 
By Stokes' law, for every small enough $\e>0$,
\[
\begin{split}
0 &= \int_{\M\setminus\Gamma_a^\e} d\nu_a  = \int_{\partial\M} \nu_a - \int_{\partial\Gamma_a^\e} \nu_a.
\end{split}
\]
Since $\nu_a$ has the same circulation as  $\beta$, 
\[
\int_{\partial\Gamma_a^\e} \nu_a =  \int_{\partial\M}\beta.
\]
%where we substituted the expression \eqref{eq:beta} for $\beta$ on the right-hand side.
Letting $\e\to0$, we obtain 
\beq
\int_{\Gamma_a}[\nu_a] =  \int_{\partial\M}\beta,
\label{eq:circ_singularity}
\eeq
where $[\nu_a]$ is the discontinuity jump of $\nu_a$ along $\Gamma_a$, whose sign is determined by the orientation of $\M$ (hence of  $\Gamma_a^\e$) and $\Gamma_a$. Note that the $1$-sided limits of $\nu_a$ at $\Gamma_a$ exist since $\nu_a$ is $C^1$-bounded. Moreover, since $\M$ is compact, the identity \eqref{eq:circ_singularity} does not depend on the choice of the metric $\g$.

Thus, the defining properties of $\nu_a$ imply that it does not satisfy the integral version \eqref{eq:closed_nu} of closedness, and as a result, must have a singularity along $\Gamma_a$.

\begin{remark}
The singular set $\Gamma_a$ of $\nu_a$ is evidently uncountable. 
Generally, if $\M$ is a compact two-dimensional manifold with or without boundary, $\Gamma$ is a submanifold of $\M$, and $\nu$ is a $C^0$-bounded closed 1-form on $\M\setminus\Gamma$, such that there exists a closed curve $C$ for which
\[
\oint_C \nu \ne 0,
\]
then $\Gamma$ cannot be a finite set.
Suppose, by contradiction that $\Gamma=\{p_1,p_2,\ldots,p_k\}$ is  finite, and assume without loss of generality that all the points in $\Gamma$ are enclosed by the curve $C$.  Assuming as above a metric $\g$, setting 
$\Gamma^\e=\cup_i B_\e(p_i)$,  and performing the same calculation,
\[
\sum_{i=1}^k \oint_{\partial B_\e(p_i)} \nu= -\oint_C \nu.
\]
If $\nu$ is bounded, then the left-hand side vanishes as $\e\to0$, yielding a contradiction. 
The physical interpretation of this observation is that \emph{there is no such thing as an edge-dislocation supported at a point} (or on a line in three dimensions).
\end{remark}

\begin{proof}[of \propref{prop:nu}]\smartqed
We construct $\nu_a$ as the differential of a discontinuous function $f$. First, define $f_0:\partial\M\to \R$ by fixing $q_0=(1,1/2)$ and letting
\[
f_0(q)=\int_{q_0}^q\beta,
\]
where the integration from $q_0$ to $q$ is along $\partial\M$ counterclockwise. If the circulation of $\beta$ is non-zero, then $f_0$ is discontinuous at $q_0$. However, its differential is well-defined and smooth at $q_0$ as it coincides with the tangential component of $\beta$. 

Next, let 
\[
\M_a = [1/2 - a/2,  1/2 + a/2] \times [0,1],
\] 
and define $\bar{f}:\M\setminus\M_a\to \R$ by integrating $\beta$ horizontally, from the boundaries inward,
\[
\bar{f}(x,y)=
\Cases{
f_0(0,y)+\int_{[(0,y),(x,y)]}\beta,\quad & x<1/2-a/2 \\
f_0(1,y)+\int_{[(1,y),(x,y)]}\beta,\quad & x>1/2+a/2.
}
\]
%where $[(0,y),(x,y)]$ for example, is the segment connecting the points $(0,y)$ and $(x,y)$. 

Denote by $p_L,p_R:\M\to\R$ the second-order Taylor expansions of $\bar{f}$ about $x_L= 1/2-a/2$ and $x_R= 1/2+a/2$ along the $x$-direction, i.e.,
\[
\begin{aligned}
p_L(x,y) = \bar{f}(x_L,y) +  \pd{\bar{f}}{x}(x_L,y) (x-x_L) + \half \pdd{\bar{f}}{x}(x_L,y) (x-x_L)^2 \\
p_R(x,y) = \bar{f}(x_R,y) +  \pd{\bar{f}}{x}(x_R,y) (x-x_R) + \half \pdd{\bar{f}}{x}(x_R,y) (x-x_R)^2.
\end{aligned}
\]

Let $r \in C^\infty(\R)$ be a monotonically-increasing function satisfying,
\[
r(t) = 0\quad \forall t\le -1/2
\Textand
r(t) = 1 \quad \forall t\ge 1/2.
\]
%For $g:\R\to\R$ denote by $p^k_{x_0}[g]$ its $k$-th order Taylor expansion around $x_0$. 
We extend $\bar{f}$ to $\M\setminus\Gamma_a$ by interpolating between $p_L$ and $p_R$, using the smooth ``connecting" function $r$ (see \figref{fig:2}),
\beq
f(x,y)=
\Cases{
\bar{f}(x,y) & |x-1/2|\geq a/2 \\
(1-r(\frac{x-1/2}{a})) p_L(x,y) 
+r(\frac{x-1/2}{a}) p_R(x,y)    & |x-1/2|<a/2.
}
\label{eq:f_cases}
\eeq

\begin{figure}
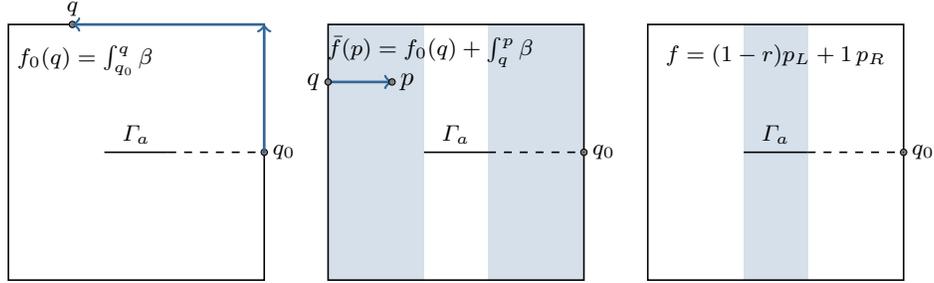

\begin{center}
\btkz[scale=1.7]
	\tkzDefPoint(-1,-1){A}
	\tkzDefPoint(1,-1){B}
	\tkzDefPoint(1,1){C}
	\tkzDefPoint(-1,1){D}
	\tkzDefPoint(1,0){O}
	\tkzDefPoint(-0.5,1){Q}
	\tkzDefPoint(-0.25,0){p}
	\tkzDefPoint(0.25,0){q}
	\tkzDrawPolygon(A,B,C,D)
	\tkzDrawPoints(O,Q)
	\tkzDrawSegment(p,q)
	\tkzDrawSegment[dashed](q,O)
	\tkzLabelSegment[above](p,q){$\Gamma_a$}
	\tkzText[right](1,0){$q_0$}
	\tkzText[above](-0.5,1){$q$}
	\tkzText[below](-0.4,0.9){$f_0(q)=\int_{q_0}^q\beta$}
	\draw[->,color=warmblue, line width = 1pt] (1,0) -- (1,1);
	\draw[->,color=warmblue, line width = 1pt] (1,1) -- (-0.5,1);
	
	\begin{scope}[xshift=2.5cm] 
	\fill[color=warmblue!20] (-1,-1) -- (-0.25,-1) -- (-0.25,1) -- (-1,1) -- cycle;
	\fill[color=warmblue!20] (1,-1) -- (0.25,-1) -- (0.25,1) -- (1,1) -- cycle;
	\tkzDefPoint(-1,-1){A}
	\tkzDefPoint(1,-1){B}
	\tkzDefPoint(1,1){C}
	\tkzDefPoint(-1,1){D}
	\tkzDefPoint(1,0){O}
	\tkzDefPoint(-0.5,1){Q}
	\tkzDefPoint(-0.25,0){p}
	\tkzDefPoint(0.25,0){q}
	\tkzDrawPolygon(A,B,C,D)
	\tkzDrawPoints(O)
	\tkzDrawSegment(p,q)
	\tkzDrawSegment[dashed](q,O)
	\tkzLabelSegment[above](p,q){$\Gamma_a$}
	\tkzText[right](1,0){$q_0$}
	\tkzDefPoint(-1,0.55){q1}
	\tkzDefPoint(-0.5,0.55){q2}
	\tkzDrawSegment[color=warmblue,->, line width=1pt](q1,q2)
	\tkzDrawPoints(q1,q2)
	\tkzLabelPoint[left](q1){$q$}
	\tkzLabelPoint[right](q2){$p$}
	\tkzText[above](-0.2,0.60){$\bar{f}(p) = f_0(q) + \int_q^p\beta$}
	\end{scope}

	\begin{scope}[xshift=5cm] 
	\fill[color=warmblue!20] (-0.25,-1) -- (0.25,-1) -- (0.25,1) -- (-0.25,1) -- cycle;
	\tkzDefPoint(-1,-1){A}
	\tkzDefPoint(1,-1){B}
	\tkzDefPoint(1,1){C}
	\tkzDefPoint(-1,1){D}
	\tkzDefPoint(1,0){O}
	\tkzDefPoint(-0.5,1){Q}
	\tkzDefPoint(-0.25,0){p}
	\tkzDefPoint(0.25,0){q}
	\tkzDrawPolygon(A,B,C,D)
	\tkzDrawPoints(O)
	\tkzDrawSegment(p,q)
	\tkzDrawSegment[dashed](q,O)
	\tkzLabelSegment[above](p,q){$\Gamma_a$}
	\tkzText[right](1,0){$q_0$}
	\tkzText[above](0,0.60){$f = (1-r)p_L + 1\, p_R$}
	\end{scope}
\etkz
\end{center}
\caption{The three stages in the construction of $f$: first $f_0$ is defined on $\partial\M$; next $\bar{f}$ is defined on the set $|x-1/2| > a/2$ by integrating the horizontal component of $\beta$ from the nearest vertical boundary; finally, $f$ is extended to the set $|x-1/2| \le a/2$ by interpolation. The dashed segment connecting $\Gamma_a$ to $q_0$ is the discontinuity line of $f$.}
\label{fig:2}
\end{figure}

We turn to evaluate $\nu_a = df$ by differentiating \eqref{eq:f_cases}.
For $x>a/2+1/2$,
\beq
\label{eq:dfcalc1}
\begin{split}
df_{(x,y)} & = \brk{\pd{}{x}\brk{\int_{[(1,y),(x,y)]}\beta}} dx + \brk{\pd{f_0}{y}(1,y)+\pd{}{y}\brk{\int_{[(1,y),(x,y)]}\beta}}dy \\
&= \beta_1(x,y) \, dx + \brk{\beta_2(1,y)+\int_1^x\pd{\beta_1}{y}(x',y)dx'}dy,
\end{split}
\eeq
where $\beta_1$ and $\beta_2$ are the components of $\beta$,
\[
\beta=\beta_1\,dx+\beta_2 \,dy.
\]
Similarly, for $x<1/2-a/2$,
\beq
df_{(x,y)}=\beta_1(x,y)\, dx + \brk{\beta_2(0,y)+\int_0^x\pd{\beta_1}{y}(x',y)dx'}dy.
\label{eq:dfcalc1b}
\eeq
While $f$ has a discontinuity along the segment $[1/2 + a, 1] \times \{1/2\}$, its one-sided derivatives along this segment  are continuous, as they are expressed in terms of the smooth 1-form $\beta$. Moreover,
\[
df|_{\Mleft} =  \beta|_{\Mleft} 
\Textand
df|_{\Mright} =  \beta|_{\Mright},
\]
proving Property~(iii). Likewise, for $|x-1/2| \ge a/2$,
\[
df(\partial_x) = \beta_1 = \beta(\partial_x),
\]
proving Property~(v).

For $(x,y)\in\M_a$, 
\beq
\label{eq:dfcalc}
\begin{split}
df_{(x,y)} &= \frac{1}{a}r'\brk{\tfrac{x-1/2}{a}}(p_R(x,y)-p_L(x,y))dx \\
&+\Brk{\brk{1-r\brk{\tfrac{x-1/2}{a}}}\pd{p_L}{x}(x,y)+r\brk{\tfrac{x-1/2}{a}} \pd{p_R}{x}(x,y)}dx\\
&+\Brk{\brk{1-r\brk{\tfrac{x-1/2}{a}}}\pd{p_L}{y}(x,y)+r\brk{\tfrac{x-1/2}{a}} \pd{p_R}{y}(x,y)}dy.
\end{split}
\eeq
The 1-form $df$ is continuous at $x= 1/2\pm a/2$, for example,
\[
\begin{split}
\lim_{x\nearrow 1/2+a/2} df(x,y) &= \pd{p_R}{x}(1/2+a/2,y)\, dx+ \pd{p_R}{y}(1/2+a/2,y)\, dy \\
&=\pd{\bar{f}}{x}(1/2+a/2,y)\, dx+\pd{\bar{f}}{y}(1/2+a/2,y)\, dy \\
&= d\bar{f}(1/2+a/2,y).
\end{split}
\]
A second differentiation shows that $\nu_a$ is continuously-differentiable at $x= 1/2\pm a/2$. 
This together with \eqref{eq:dfcalc} proves Property~(i) and consequently also Property~(ii). 

It remains to {prove Property~(iv), that $df$ and $\beta$ have the same circulations. 
This follows from our construction of $f_0$ on $\partial\M$, 
\[
\begin{split}
\int_{\partial\M} df &= \lim_{\e\to0} \brk{f(1,1/2-\e) - f(1,1/2+\e)} \\
&= \lim_{\e\to0} \brk{f_0(1,1/2-\e) - f_0(1,1/2+\e)} \\
&= \int_{\partial\M} \beta.
\end{split}
\]
}

\qed\end{proof}
%%%%%%
%\Elik{
%\begin{remark}
%Some details concerning gluing conditions.. to complete.  Maybe not necessary.
%\end{remark}
%}

The 1-form $\nu_a$ (which is only defined on $\M\setminus\Gamma_a$) induces a 1-current on $\M$,
\[
T_{\nu_a}(\alpha)=\int_\M \nu_a\wedge \alpha
\qquad 
\alpha\in\Omega_c^1(\M).
\]
Its boundary  is the 0-current,
\[
\partial T_{\nu_a}(f) = T_{\nu_a}(df) = \int_\M \nu_a\wedge df
\qquad 
f\in C_c^\infty(\M).
\]
Integrating by parts, we obtain
\[
\partial T_{\nu_a}(f) =\int_{\Gamma_a} f [\nu_a],
\]
where for $|x-1/2| < a/2$,
\[
\begin{split}
[\nu_a](x) &= \lim_{\e\to0} \brk{df(x,1/2+\e) - df(x,1/2-\e)} \\
&=\frac{1}{a}r'\brk{\tfrac{x-1/2}{a}} \,  \lim_{\e\to0} \brk{p_R(x,1/2+\e)- p_R(x,1/2-\e)} \\
&=  \frac{1}{a}r'\brk{\frac{x-1/2}{a}}\, \int_{\partial\M} \beta.
\end{split}
\]

To conclude, we view $\nu_a$ as a layering form on $\M$ having an edge-dislocation concentrated on the hyper-surface $\Gamma_a$. The locus of the dislocation is revealed by the boundary of the differential current induced by $\nu_a$. Note that $\M\setminus\Gamma_a$ is defect-free only to the extent detectable by $\nu_a$. Generally, $\M\setminus\Gamma_a$ may contain defects detected by other layering forms.
%\\ \\
%\Elik{
%For later purpose  we show that the construction above may be generalised for a general (non-closed) 1-form.
%\begin{corollary}
%\label{cor:generalform}
%Let $\phi\in\Omega^1(\R^2)$ and $\M=[0,1]^2\subset\R^2$. Then for every $a>0$ there exists a continuously differential 1-form $\phi_a$ on $\M\setminus\Gamma_a$  which is $C^1$ bounded, closed, and coincides with $\phi$ on $\partial\M$. 
%\end{corollary}
%\begin{proof}
%Denote by $b=\int_{\partial\M}\phi$, the circulation of $\phi$ and consider the 1-form $\beta=(1+bx)dy$ on $\M$. Proposition \ref{prop:nu} gives a closed 1-form $\nu_a$ on $\M\setminus \Gamma_a$ which coincides with $\beta$ on $\partial\M$. The difference $\delta:=\phi-\nu_a$ satisfies $\int_{\partial\M}\delta=0$. Lemma \ref{lem:correction} then gives a closed 1-form $\omega\in\Omega^1(\M)$  satisfying $\omega|_{\partial\M}=\delta$. Thus $\phi_a=\nu_a+\omega$ satisfies the desired properties. 
%\end{proof}
%}
%%%%%%%%%%
%%%%%%%%%%%

%%%%%%%%%%%%%%%%%%%%%%%%%%%
\section{Homogenization of distributed edge-dislocations}
\label{sec:homogenization}

We proceed to construct a singular layering form corresponding to an $n$-by-$n$ array of edge-dislocations, each of magnitude of order $1/n^2$, using \propref{prop:nu} as a building block.

For $(x_0,y_0)\in\R^2$, denote by $\tau_{(x_0,y_0)}:\R^2\to\R^2$ the translation operator
\[
\tau_{(x_0,y_0)}(x,y) = (x+x_0,y+y_0).
\]
Likewise, for $\lambda>0$, denote by $S_\lambda:\R^2\to \R^2$ the scaling operator
\[
S_\lambda(x,y) = (\lambda x,\lambda y).
\] 

Let $n\in\bbN$ be given; for every $0\le k,j <  n$, let
\[
\calM^{(n)}_{kj} = S_{1/n} \circ \tau_{(k,j)}  (\M)
\]
be translated and rescaled copies of $\M$, forming an $n$-by-$n$ tiling of $\M$. 
By construction, 
\beq
\label{eq:iota}
\iota_{kj}^{(n)} =   S_{1/n} \circ \tau_{(k,j)} : \M \to \M^{(n)}_{kj} 
\eeq
is a diffeomorphism (see \figref{fig:n-by-n}).
Similarly, let
\[
\Gamma^{(n)}_{kj} = \iota_{kj}^{(n)}(\Gamma_{a/n})
\]
be segments of lengths $a/n^2$ located at the centers of each square.
Finally, denote by
\[
\Gamma^{(n)} =  \bigcup_{k,j=0}^{n-1} \Gamma^{(n)}_{jk},
\]
the union of those segments and note that $|\Gamma^{(n)}|=a$.

\begin{figure}
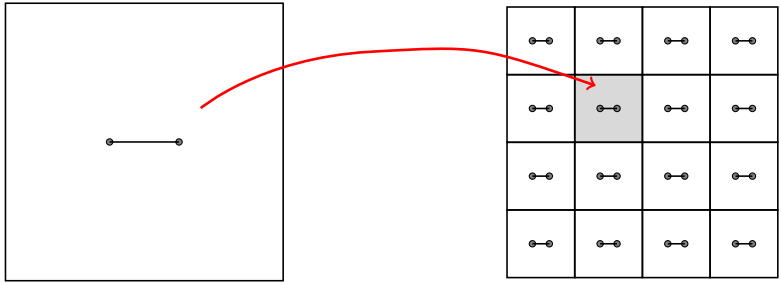

\begin{center}
\btkz[scale=1.5]
	\begin{scope}[yshift = 0.2cm, scale=1.23]	
	\tkzDefPoint(-1,-1){A}
	\tkzDefPoint(1,-1){B}
	\tkzDefPoint(1,1){C}
	\tkzDefPoint(-1,1){D}
	\tkzDefPoint(1,0){O}
	\tkzDefPoint(-0.25,0){p}
	\tkzDefPoint(0.25,0){q}
	\tkzDrawPolygon(A,B,C,D)
	\tkzDrawPoints(p,q)
	\tkzDrawSegment(p,q)
	\end{scope}

	\foreach \x in {1} {
	\foreach \y in {2}{
	\begin{scope}[xshift=100+0.6*\x cm, yshift = -20 + 0.6*\y cm, scale=0.3]
		\tkzDefPoint(-1,-1){A}
		\tkzDefPoint(1,-1){B}
		\tkzDefPoint(1,1){C}
		\tkzDefPoint(-1,1){D}
		\tkzDefPoint(1,0){O}
		\tkzFillPolygon[gray!30](A,B,C,D)
	\end{scope}
	}}
	
	\foreach \x in {0,1,2,3} {
	\foreach \y in {0,1,2,3}{
	\begin{scope}[xshift=100+0.6*\x cm, yshift = -20 + 0.6*\y cm, scale=0.3]
		\tkzDefPoint(-1,-1){A}
		\tkzDefPoint(1,-1){B}
		\tkzDefPoint(1,1){C}
		\tkzDefPoint(-1,1){D}
		\tkzDefPoint(1,0){O}
		\tkzDefPoint(-0.25,0){p}
		\tkzDefPoint(0.25,0){q}
		\tkzDrawPolygon(A,B,C,D)
		\tkzDrawPoints(p,q)
		\tkzDrawSegment(p,q)
	\end{scope}
	}}
	
	\draw [<-, line width=1pt, red] plot [smooth, tension=1] coordinates {(4,0.7) (2,1) (0.5,0.5)};
\etkz
\end{center}
\caption{The diffeomorphism $\iota_{kj}^{(n)}$ for $n=4$, $k=1$ and $j=2$}
\label{fig:n-by-n}
\end{figure}
%For $0\le k,j\leq n-1$, let 
%\[
%c^{(n)}_{kj}=(x^{(n)}_k,y^{(n)}_j)= \brk{\frac{k}{n},\frac{j}{n}}
%\]
%be the lower-left vertices of the squares
%\[
%\calM^{(n)}_{kj} = c^{(n)}_{kj}  +  \Brk{0,\frac{1}{n}}^2,
%\] 
%and let
%\[
%\Gamma^{(n)}_{kj} = c^{(n)}_{kj} + \brk{\frac{1}{2n},\frac{1}{2n}} +  \Brk{-\frac{a}{2n^2},\frac{a}{2n^2}} \times\{0\}
%\]
%be segments of lengths $a/n^2$ located at the centers of each square.

Let $\beta\in\Omega^1(\M)$ be a layering form. We approximate it by a sequence of singular layering forms, 
\[
\nu^{(n)} \in \Omega^1(\M\setminus\Gamma^{(n)}).
\]

Let 
\beq
\beta^{(n)}_{kj}=(\iota^{(n)}_{kj})^\star \beta|_{\M^{(n)}_{kj}}\in \Omega^1(\M),
\label{eq:beta_nkj}
\eeq
be the pullback of $\beta$ (restricted to $\M^{(n)}_{kj}$) to $\M$ and let $\mu^{(n)}_{kj}\in\Omega^1(\M\setminus\Gamma_{a/n})$ be the singular 1-form defined in \propref{prop:nu}, with $\beta^{(n)}_{kj}$ playing the role of $\beta$. Pushing forward into $\M^{(n)}_{kj}$, we set
%By definition we then have for every $x,y\in [0,1]$.
%\beq
%\label{eq:betan}
%\beta^{(n)}_{kj}(x,1)=\beta^{(n)}_{k(j+1)}(x,0)\textand \beta^{(n)}_{kj}(1,y)=\beta^{(n)}_{(k+1)j}(0,y).
%\eeq
%This also implies 
\beq
\label{eq:local_nu_n}
\nu^{(n)}|_{\M^{(n)}_{kj}} = (\iota_{kj}^{(n)})_\star\mu_{kj}^{(n)}.
\eeq

%We construct $\nu^{(n)}$ in each of the squares $\M^{(n)}_{kj}$ separately. 
%with inverse
%\[
%(\iota_{kj}^{(n)})^{-1}= S_{1/n} \circ \tau_{(k,j)} : \M\to \M^{(n)}_{kj} .
%\] 

%We then note that the ``affine" form $\beta$ satisfies the identity,
%\beq
%\beta|_{\M^{(n)}_{kj}} = (1+bx_k) \, dy|_{\M^{(n)}_{kj}} + \frac{1}{n^2} {\iota_{kj}^{(n)}}^\star(\beta - dy),
%\label{eq:local_beta}
%\eeq
%i.e.,  it is equal to its own pullback up to translation and scaling. Inspired by this identity, we define
%\beq
%\nu^{(n)}|_{\M^{(n)}_{kj}} = (1+bx_k) \, dy|_{\M^{(n)}_{kj}} + \frac{1}{n^2} {\iota_{kj}^{(n)}}^\star(\nu_{a/n} - dy).
%\label{eq:local_nu_n}
%\eeq

%%%%%
%%%%%
%%%%%
%%%%%

%\Elik{
%Note that in the construction of $\nu_a=df$ in proposition \ref{prop:nu} we had the freedom of defining the boundary values $\bar{f}|_{\partial\M}=f_0$ (for $|x-1/2|>a/2$) up to a constant. We will therefore use this freedom carefully to make $\nu^{(n)}$  smooth on the boundaries  $\partial \M^{(n)}_{kj}$. 
%}

%%%%%%%%%%
\begin{proposition}
Equation \eqref{eq:local_nu_n} for $0\le k,j<n$ defines a 
1-form $\nu^{(n)}$ on $\M$, satisfying 
\begin{enumerate}[label=(\roman*)]
\item $\nu^{(n)}$ is $C^1$-bounded.
\item $\nu^{(n)}$ is closed.
\item $\nu^{(n)}$ has the same circulation as $\beta$ in each sub-domain: for every $0\leq k,j\leq n-1$,
\[
\int_{\partial\M^{(n)}_{kj}}\nu^{(n)}=\int_{\partial\M^{(n)}_{kj}}\beta.
\]
\item $\nu^{(n)}$ coincides with $\beta$ on the vertical segments $L_k=\{\frac{k}{n} \}\times [0,1]$ for $0\leq k\leq n$. 
\end{enumerate}
\end{proposition}
%%%%%%%%%

%%%%%%%
\begin{proof}\smartqed
We first show that $\nu^{(n)}$ is well-defined and satisfies Property~(i). It is obviously smooth in the interior of each $\M^{(n)}_{kj}\setminus \Gamma^{(n)}_{kj}$. It remains to prove that it is continuously-differentiable on the ``skeleton" $\cup_{k,j}\partial\M^{(n)}_{kj}$. Note that
\[
\partial\M^{(n)}_{kj} = \iota^{(n)}_{kj}(\Mleft) \cup \iota^{(n)}_{kj}(\Mright) \cup \iota^{(n)}_{kj}(\Mtop) \cup \iota^{(n)}_{kj}(\Mbottom). 
\]

By \eqref{eq:beta_nkj}, since the diffeomorphism $\iota^{(n)}_{kj}$ is a combination of a translation and a scaling,  
\[
\beta^{(n)}_{kj}(\partial_x) =\frac{1}{n} \beta(\partial_x) \circ \iota^{(n)}_{kj}
\Textand
\beta^{(n)}_{kj}(\partial_y) =\frac{1}{n} \beta(\partial_y) \circ \iota^{(n)}_{kj},
\]
which are equalities between functions on $\M$. In particular, for every $x,y\in[0,1]$, and $v \in \{\partial_x,\partial_y\}$
\[
\begin{gathered}
\beta^{(n)}_{k\,j+1}(v)(x,0) = \beta^{(n)}_{kj}(v)(x,1) \\
\beta^{(n)}_{k+1\,j}(v)(0,y) = \beta^{(n)}_{kj}(v)(1,y) \\
\end{gathered}
\]
By the same argument, for $w\in \{\partial_x,\partial_y\}$
\[
\begin{gathered}
\calL_w \beta^{(n)}_{k\,j+1}(v)(x,0) = \calL_w \beta^{(n)}_{kj}(v)(x,1) \\
\calL_w\beta^{(n)}_{k+1\,j}(v)(0,y) = \calL_w\beta^{(n)}_{kj}(v)(1,y).
\end{gathered}
\]
By \eqref{eq:dfcalc1}, \eqref{eq:dfcalc1b} and \eqref{eq:dfcalc}, the construction of $\mu^{(n)}_{kj}$ only depends on $\beta^{(n)}_{kj}$ (and the smooth function $r$). Moreover, $\mu_{kj}^{(n)}$ and its derivative on every side of $\partial\M$ depend only on $\beta^{(n)}_{kj}$ and its derivatives on that side.  As a result, for every $x,y\in[0,1]$, and $v,w = \{\partial_x,\partial_y\}$,
\[
\begin{gathered}
\mu^{(n)}_{k\,j+1}(v)(x,0) = \mu^{(n)}_{kj}(v)(x,1) \\
\mu^{(n)}_{k+1\,j}(v)(0,y) = \mu^{(n)}_{kj}(v)(1,y) \\
\calL_w \mu^{(n)}_{k\,j+1}(v)(x,0) = \calL_w \mu^{(n)}_{kj}(v)(x,1) \\
\calL_w\mu^{(n)}_{k+1\,j}(v)(0,y) = \calL_w\mu^{(n)}_{kj}(v)(1,y).
\end{gathered}
\]
Since the relation between $ \mu^{(n)}_{kj}$ and $\nu^{(n)}$ is once again a pullback under a combination of scaling and translation, we obtain that $\nu^{(n)}$ is continuously-differentiable along the skeleton.

We proceed to prove Property~(iv): by Property~(iii) of \propref{prop:nu},  
\[
\begin{split}
\nu^{(n)}|_{\iota^{(n)}_{kj}(\Mleft)} &= (\iota_{kj}^{(n)})_\star\mu_{kj}^{(n)}|_{\iota^{(n)}_{kj}(\Mleft)} \\
&= (\iota_{kj}^{(n)})_\star\beta_{kj}^{(n)}|_{\iota^{(n)}_{kj}(\Mleft)} \\ 
&= (\iota_{kj}^{(n)})_\star(\iota^{(n)}_{kj})^\star \beta|_{\iota^{(n)}_{kj}(\Mleft)} \\
&= \beta|_{\iota^{(n)}_{kj}(\Mleft)},
\end{split}
\]
i.e., $\nu^{(n)}$ coincides with $\beta$ on the vertical components of the skeleton.

Property~(ii) is immediate as $\mu^{(n)}_{kj}$ are closed and closedness is invariant under the pullback operation. Finally, Property~(iii) follows from Property~(iv) in \propref{prop:nu},
\[
\begin{split}
\int_{\partial\M^{(n)}_{kj}}\nu^{(n)} &= \int_{\iota^{(n)}_{kj}(\partial\M)} ((\iota_{kj}^{(n)})^{-1})^\star \mu_{kj}^{(n)} \\
&= \int_{\partial\M}\mu_{kj}^{(n)} \\
&= \int_{\partial\M}\beta_{kj}^{(n)} \\
%&= \int_{(\iota_{kj}^{(n)})^{-1}(\partial\M^{(n)}_{kj})} (\iota_{kj}^{(n)})^\star \beta \\
&= \int_{\partial\M^{(n)}_{kj}}\beta.
\end{split}
\]
{}

\qed\end{proof}
%%%%%%

As in the case of a single dislocation, we define for each $n$ the 1-current induced by $\nu^{(n)}$: 
\[
T_{\nu^{(n)}}(\alpha)=\int_\M \nu^{(n)}\wedge \alpha
\qquad
\alpha \in \Omega^1_c(\M).
\]
Its boundary $\partial T_{\nu^{(n)}}$ is a 0-current given by 
\[
\partial T_{\nu^{(n)}}(f)=\sum_{k,j=1}^{n-1}\int_{\Gamma^{(n)}_{kj}}f[\nu^{(n)}]_{\Gamma^{(n)}_{kj}}
\]
where $[\nu^{(n)}]_{\Gamma^{(n)}_{kj}}$ is the discontinuity jump of $\nu^{(n)}$ along $\Gamma^{(n)}_{kj}$, given by,
\[
[\nu^{(n)}]_{\Gamma^{(n)}_{kj}}(x,(j+1/2)/n) =
 \frac{n}{a}r'\brk{\frac{nx - k -1/2}{a}}\, \int_{\partial\M^{(n)}_{kj}} \beta.
\]
Thus, we view $\nu^{(n)}$ as a layering form on $\M$ having $n^2$ edge-dislocations concentrated on $\Gamma^{(n)}$. The loci of the dislocations are revealed by the boundary of the differential current induced by $\nu^{(n)}$. Here too, $\M\setminus\Gamma^{(n)}$ is defect-free only to the extent detectable by $\nu^{(n)}$.

%%%%%%
\begin{theorem}[Homogenization]
\label{thm:Homogenization}
The sequence 
$T_{\nu^{(n)}}$ of 1-forms converges to $T_\beta$ in the sense of currents: for every $\alpha\in \Omega^1_c(\M)$,
\[
\limn \int_\M \nu^{(n)}\wedge \alpha =  \int_\M \beta\wedge \alpha,
\]
or equivalently,
\beq
\label{hom:2}
\limn T_{\nu^{(n)}-\beta}(\alpha)= 0.
\eeq
\end{theorem}
%%%%%%

%%%%%%
\begin{proof}\smartqed

Choose any metric on $\M$; for concreteness we will take the Euclidean metric associated with the parametrization.
%Let $\e>0$ be given, and let $N_0\in\bbN$ be such that 
%\[
%\|\beta_{{x,y}}-\beta_{(x',y')}\| < \e
%\qquad
%\text{whenever}
%\qquad
%|(x-x',y-y')| < 1/N_0.
%\]  
By our choice of metric, if $\beta = \beta_1 \, dx + \beta_2\, dy$, then
\[
\|\beta_{(x,y)}\|^2 = \beta_1^2(x,y) + \beta_2^2(x,y).
\]

For every $\alpha\in \Omega^1_c(\M)$,
\[
\begin{split}
T_{\nu^{(n)}-\beta}(\alpha) &= 
\sum_{k,j=0}^{n-1} \int_{\M^{(n)}_{kj}} (\nu^{(n)} -\beta)\wedge \alpha \\
&= \sum_{k,j=0}^{n-1} \int_{\iota_{kj}^{(n)}(\M)} ((\iota_{kj}^{(n)})^{-1})^\star(\mu^{(n)}_{kj}-\beta^{(n)}_{kj})\wedge \alpha \\
&= \sum_{k,j=0}^{n-1}\int_\M  (\mu^{(n)}_{kj}-\beta^{(n)}_{kj})\wedge (\iota_{kj}^{(n)})^\star\alpha.
\end{split}
\]
Fix $0\leq k,j\leq n-1$.  Since
\[
\Norm{{(\iota_{kj}^{(n)})^\star \alpha|_{\M^{(n)}_{kj}}}}_\infty \le 
\frac{1}{n} \|\alpha\|_\infty,
\]
it follows that
\[
\begin{split}
\left| \int_\M  (\mu^{(n)}_{kj}-\beta^{(n)}_{kj})\wedge (\iota_{kj}^{(n)})^\star\,\alpha\right| 
&\le \frac{1}{n} \|\alpha\|_\infty  \sup_{\|\xi\|_\infty=1} \Abs{\int_\M (\mu^{(n)}_{kj}-\beta^{(n)}_{kj}) \wedge \xi} \\
&\le \frac{1}{n}\, \|\alpha\|_\infty \int_\M |\mu^{(n)}_{kj}-\beta^{(n)}_{kj}| \, dx\wedge dy
\end{split}.
\]
Thus, so far,
\[
\begin{split}
T_{\nu^{(n)}-\beta}(\alpha) \le n\, \|\alpha\|_\infty \sup_{0\le k,j< n} \int_\M |\mu^{(n)}_{kj}-\beta^{(n)}_{kj}| \, dx\wedge dy.
\end{split}
\]

Now,
\[
(\beta^{(n)}_{kj})_{(x,y)} = \frac{1}{n}\beta_1\brk{\frac{x+k}{n},\frac{y+j}{n}} \,dx + \frac{1}{n}\beta_2\brk{\frac{x+k}{n},\frac{y+j}{n}} \,dy.
\]
By \eqref{eq:dfcalc1b}, for $x<1/2-a/2n$, 
\[
\begin{split}
(\mu^{(n)}_{kj})_{(x,y)} &= \frac{1}{n}\beta_1\brk{\frac{x+k}{n},\frac{y+j}{n}}\, dx \\ 
&+ \brk{\frac{1}{n}\beta_2\brk{\frac{k}{n},\frac{y+j}{n}} 
+ \int_0^x\frac{1}{n^2}\pd{\beta_1}{y}\brk{\frac{x'+k}{n},\frac{y+j}{n}} dx'}dy,
\end{split}
\]
so that 
\[
\begin{split}
n\, |\mu^{(n)}_{kj}-\beta^{(n)}_{kj}|(x,y) &\le 
\Abs{\beta_2\brk{\frac{x+k}{n},\frac{y+j}{n}}  - \beta_2\brk{\frac{k}{n},\frac{y+j}{n}} } \\
&\qquad+ \frac{1}{n} \int_0^x\Abs{\pd{\beta_1}{y}\brk{\frac{x'+k}{n},\frac{y+j}{n}}} dx' \\
&\le \frac{1}{n} \brk{\Norm{\pd{\beta_2}{x}}_\infty + \Norm{\pd{\beta_1}{y}}_\infty}.
\end{split}
\]

The same bound is obtained for $x>1/2+a/2n$. Finally, for $|x-1/2|<a/2n$, using \eqref{eq:dfcalc}, and noting that $p_L$ and $p_R$ are $O(1/n)$, we obtain that
\[
n\, |\mu^{(n)}_{kj}-\beta^{(n)}_{kj}|(x,y) \le \frac{C}{a} \|r'(x)\|_\infty,
\]
where $C$ is some constant.
Putting it all together,
\[
\begin{split}
T_{\nu^{(n)}-\beta}(\alpha) &\le n\, \|\alpha\|_\infty \sup_{0\le k,j< n} \int_{\M\setminus\M_{a/n}} |\mu^{(n)}_{kj}-\beta^{(n)}_{kj}| \, dx\wedge dy \\
&\qquad+ n\, \|\alpha\|_\infty \sup_{0\le k,j< n} \int_{\M_{a/n}} |\mu^{(n)}_{kj}-\beta^{(n)}_{kj}| \, dx\wedge dy \\
&\le \frac{\|\alpha\|_\infty}{n}\brk{\Norm{\pd{\beta_2}{x}}_\infty + \Norm{\pd{\beta_1}{y}}_\infty + C\,\|r'(x)\|_\infty }.
\end{split}
\]
Letting $n\to\infty$ we obtain the desired result.

%\Elik{
%We now turn to evaluate (pointwise) the $x$ and $y$ components of $\mu^{(n)}_{kj}-\beta^{(n)}_{kj}$ on $\M_{a/n}= \Gamma_{a/n}\times [0,1]$ and $\M\setminus M_{a/n}$. 
%}
%By \eqref{eq:nu_n_minus_beta},
%\[
%T_{\nu^{(n)}-\beta}(\alpha) = 
% \frac{1}{n^2} \sum_{k,j=0}^{n-1} \int_{\M^{(n)}_{kj}}  {\iota_{kj}^{(n)}}^\star(\nu_{a/n} - \beta) \wedge \alpha|_{\M^{(n)}_{kj}}.
%\]
%Changing variables,
%\[
%T_{\nu^{(n)}-\beta}(\alpha) = 
%\frac{1}{n^2} \sum_{k,j=0}^{n-1} \int_{\M}  (\nu_{a/n} - \beta) \wedge 
%{(\iota_{kj}^{(n)})^{-1}}^\star \alpha|_{\M^{(n)}_{kj}}.
%\]
%
%By the properties of the pullback,
%\[
%\Norm{{(\iota_{kj}^{(n)})^{-1}}^\star \alpha|_{\M^{(n)}_{kj}}}_\infty \le 
%\frac{1}{n} \|\alpha\|_\infty.
%\]
%
%Thus
%\[
%|T_{\nu^{(n)}-\beta}(\alpha)| \le \frac{\|\alpha\|_\infty}{n} \sup_{\|\xi\|_\infty=1} \Abs{\int_\M (\nu_{a/n} - \beta) \wedge \xi}.
%\]
%
%The infinity norm of $\nu_{a/n}$ is $O(n)$, however, it is $O(1)$, everywhere except for a domain of area  $O(1/n)$ (once again, this is a statement independent of a choice of a metric on $\M$).  Hence,
%$T_{\nu^{(n)}-\beta}(\alpha) = O(1/n)$, which completes the proof. 
\qed\end{proof}

%%%%%%%%%%%%
\section{Singular torsion and its homogenization}
\label{sec:torsion}

Thus far, we analyzed a lattice structure through a single layering form, representing a single family of Bravais surfaces. 
In $n$ dimension, a lattice structure is fully determined by a set of $n$ linearly-independent layering forms, i.e., by a coframe $\{\cof{i}\}$. Denote by $\{e_i\}$ the frame field dual to $\{\cof{i}\}$.

A frame-coframe structure induces a path-independent parallel transport,
\beq
\Pi_p^q : T_p\M\to T_q\M
\qquad\text{given by}\qquad
\Pi_p^q = e_i|_q \otimes \cof{i}|_p.
\label{eq:Pi}
\eeq
In turn, the specification of a path-independent parallel transport induces a connection $\nabla$ having trivial holonomy, which locally implies zero curvature. 
By construction, the frame field $\{e_i\}$ and its dual $\{\cof{i}\}$ are $\nabla$-parallel sections,
\[
\nabla e_i =0
\Textand
\nabla \cof{i} = 0.
\]

The torsion tensor associated with $\nabla$ is a $T\M$-valued 2-form $\tau$, given by
\[
\tau(e_i,e_j)=\nabla_{e_i}e_j-\nabla_{e_j}e_i-[e_i,e_j]=[e_j,e_i].
\]
Since for every $1\leq i,j,k\leq n$,
\[
\begin{split}
d\cof{i}(e_j,e_k) &=e_j (\cof{i}(e_k)) - e_k (\cof{i}(e_j)) - \cof{i}([e_j,e_k]) \\
&=\cof{i}([e_k,e_j]) \\
&= \cof{i}(\tau(e_j,e_k)),
\end{split}
\]
we conclude that $d\cof{i} = \cof{i}\circ\tau$,  or equivalently,
\beq
\tau = e_i\otimes d\cof{i}.
\label{eq:torsion}
\eeq
In particular, torsion vanishes if and only if $d\cof{i}=0$ for all $1\leq i\leq n$, or  equivalently, if $[e_i,e_j]=0$ for all $1\leq i,j\leq n$.

%For example, taking in two-dimensions
%\[
%\cof{1} = (1+bx)\, dy
%\Textand
%\cof{2} = dx,
%\]
%the frame field dual to $\{\cof1,\cof2\}$ is
%\[
%e_1 = (1+bx)^{-1}\partial_y
%\Textand
%e_2 = \partial_x.
%\] 
%The parallel transport operator is 
%\[
%\Pi_{(x,y)}^{(x',y')}(c_1\partial_x + c_2\partial_y) = 
%c_1\, \partial_x + c_2 \frac{1+bx}{1+bx'}\, \partial_y,
%\]
%and the torsion tensor is given by 
%\beq
%\tau = e_1\otimes d\cof{1} = \frac{b}{1+bx} \, \partial_y \otimes  dx \wedge dy.
%\label{eq:torsion_example}
%\eeq

The question we are addressing henceforth is in what sense may the smooth torsion $\tau$ given by \eqref{eq:torsion} a limit of torsions associated with singular dislocations. 
For example, let $\calM$, $\beta$ and $\nu^{(n)}$ be defined as in the previous section, and suppose that 
\[
\cof{1}_{(n)} = \nu^{(n)}
\Textand
\cof{2}_{(n)} = dx
\]
is a sequence of coframe fields (namely, $\nu^{(n)}$ are $dx$ are linearly independent).
By the analysis of the previous section (and trivially for $\cof2$),
\[
\limn T_{\cof{1}_{(n)}} = T_\beta
\Textand
\limn T_{\cof{2}_{(n)}} = T_{dx},
\]
i..e,
\[
\limn \{\cof{1}_{(n)},\cof{2}_{(n)}\} = \{\beta,dx\}
\]
in the sense of weak convergence of currents.

Since the coframe field $\{\cof{1}_{(n)},\cof{2}_{(n)}\}$ consists of closed forms, the induced torsion on $\M\setminus\Gamma^{(n)}$ vanishes identically for every $n$, 
\[
\tau^{(n)} = e_i^{(n)}\otimes d\cof{i}_{(n)} = 0,
\]
which, if $d\beta\ne0$,  does not converge to the torsion 
\[
\tau =  \frac{1}{\beta_2} \partial_y\otimes d\beta
\]
associated with the limiting coframe field in any classical sense.

The question is how to cast a weak convergence of torsion in the framework of de-Rham currents. Torsion is a tangent bundle-valued 1-form. While it is possible to define currents associated with tangent bundle-valued forms, see e.g. \cite{RS12}, this approach doesn't seem applicable here. A simple heuristic argument shows that if we try to interpret torsion as a distribution for a discontinuous coframe field, we obtain the product of a discontinuous section $e_i$ and the derivative of a discontinuous section $d\cof{i}$, which is not well-defined.  

A hint toward a correct interpretation of singular torsion is obtained by considering Burgers circuits: Let $C$ be a simple, oriented, regular closed curve in $\M$. The Burgers vector associated with the curve $C$ is a parallel vector field $B$ \cite{KMS15}, whose value at a reference point $p$ is given by
\[
B_p =\oint_C \Pi^p_{\gamma}(d\gamma),
\]
where $\Pi^p$ is the parallel-transport to $p$,
given by
\[
\Pi^p = e_i|_p\otimes \cof{i},
\]
and $\gamma$ is a parametrization for $C$. Interpreting $\Pi^p$ as a $T_p\M$-valued 1-form, we rewrite the Burgers vector $B_p$  in a more abstract form,
\[
B_p=\oint_C \Pi^p.
\]
Applying Stokes' theorem,
\[
B_p = \int_\Sigma d\Pi^p,
\]
where $\partial\Sigma=C$. Hence,
\[
B_p =  e_i|_p \int_\Sigma d\cof{i}.
\]
Thus, having chosen a reference point $p$, the Burgers vector for a loop $C$ is an integral over the area enclosed by this loop of a Burgers vector density
\[
e_i|_p \otimes d\cof{i},
\]
which is a $T_p\M$-valued 2-form; it is nothing but the torsion $\tau$, whose output, once acting on a bivector, is parallel-transported to the reference point $p$. We henceforth denote 
\[
\tau_p = \Pi^p\circ\tau = e_i|_p \otimes d\cof{i}.
\]
The notion of singular torsion may now be easily defined as the distributional counterpart of $\tau_p$ by replacing $d\cof{i}$ with the boundary current $\partial T_{\cof{i}}$.  However, we first need to define the notion of a singular frame. 
Rather than choosing the most general framework possible, we adopt a possibly restrictive but yet sufficiently rich and  physically motivated approach:

\begin{definition}
Let $\M$ be a compact $n$-dimensional manifold. A collection $\{\cof{i}\}_{i=1}^n$ of 1-forms is called a \Emph{singular coframe} for $\M$ if for every $1\leq i\leq n$,  there exists a compact $(n-1)$-dimensional submanifold $\Gamma^i\subset\M$, such that
\begin{enumerate}
\item Each $\cof{i}$ is a $C^1$-bounded 1-form on $\M\setminus \Gamma^i$.
\item $\{\cof{i}_p\}_{i=1}^n$ is a basis for $T_p^*\M$ for every $p\in \M\setminus \Gamma$ where $\Gamma=\cup_i\Gamma^i$.
\item $ \M\setminus \Gamma$ is path connected and $\partial\M\cap\Gamma=\emptyset$.
\end{enumerate}
A \Emph{closed singular coframe} is a singular coframe  $\{\cof{i}\}_{i=1}^n$ satisfying $d\cof{i}=0$ on $\M\setminus\Gamma^i$ for every $1\leq i\leq n$. 
\end{definition}

Recall that if a layering form $\omega\in \Omega^1(\M)$ is closed, its induced layering structure (foliation) is defect free. A closed singular coframe therefore corresponds to isolated defects which are concentrated on a set of measure zero.

We next define singular torsion: 

\begin{definition}
Let $\{\cof{i}\}_{i=1}^{n}$ be a \emph{singular coframe field} on $\M$ and let $p\in\M\setminus\Gamma$ be an arbitrary reference point.  The \emph{torsion current}, is a $T_p\M$-valued $(n-2)$-current given by,
\[
\calT =  e_i|_p \, \partial T_{\cof{i}}.
\] 
\end{definition}

First, note that for a smooth coframe $\{\cof{i}\}_{i=1}^{n}$, the torsion current is given by
\beq
\label{eq:smoothtorsion}
\calT(\alpha)=e_i|_p\,\partial T_{\cof{i}}(\alpha)=e_i|_p T_{d\cof{i}}(\alpha)=T_{\tau_p}(\alpha),\quad \alpha\in \Omega^{n-1}_c(\M).
\eeq
In other words, in the smooth case, the torsion current $\calT$ is the $T_p\M$-valued $(n-2)$-current induced by the smooth $T_p\M$-valued 2-form $\tau_p$. 

In the case of a closed singular coframe (isolated defects),  the singular torsion is supported on the singularity hyper-surfaces $\{\Gamma^i\}$ and is given explicitly by
\beq
\label{eq:singulartorsion}
\mathcal{T}[p](\eta)=\sum_{i=1}^n\brk{\int_{\Gamma^i}[\cof{i}]_{\Gamma^i}\wedge \eta}e_i(p),
\eeq	
where $[\cof{i}]_{\Gamma^i}$ is the discontinuity jump of $\cof{i}$ along $\Gamma^i$ and $\eta\in\Omega_c^{n-2}(\M)$.
For a general (non-closed) singular frame $\{\cof{i}\}$, the torsion current naturally decomposes to a smooth component  as in equation \eqref{eq:smoothtorsion} and a singular component as in \eqref{eq:singulartorsion}.

We have thus obtained the following corollary:
%%%%%%%%%%%%%
\begin{corollary}[Homogenization of torsion]
\label{corr:torsion}
Let $\{\cof{i}_{(k)}\}$ be a sequence of (possibly) singular coframes and $p\in\M$,  a reference point, satisfying:
\begin{enumerate}
\item   There exists a (possibly) singular frame $\{\cof{i}\}$ such that $\{\cof{i}_{(k)}\}$ converges to $\{\cof{i}\}$ in the sense of currents. That is 
\[
T_{\cof{i}_{(k)}}\to T_{\cof{i}}\quad\text{as}\,\,k\to\infty,\quad \forall \,1\leq i\leq n.
\]
\item The point $p$ is outside the singularity sets of $\{\cof{i}_{(k)}\}$ and $\{\cof{i}\}$ and $(\cof{i}_{(k)})_p\to \cof{i}_p$ (pointwise) for every $1\leq i\leq n$.
\end{enumerate}
%converging to $\{\cof{i}\}$ in the sense of currents \Elik{and let $p\in \M$ outside the singular sets of $\{\cof{i}_{(n)}\}$ for which }. 
Let 
\[
\calT_k =  e_i^{(k)}|_p \, \partial T_{\cof{i}_{(k)}}
\Textand
\calT =  e_i|_p \, \partial T_{\cof{i}}
\]
be the corresponding $T_p\M$-valued $(n-2)$-torsion currents. Then, $\calT_k\to\calT$ in the sense of currents.
\end{corollary}

In particular, if $\{\cof{i}_{(k)}\}$ are singular closed frames for every $k$ and the limiting frame $\{\cof{i}\}$ is smooth, then  $\mathcal{T}_k$  and $\mathcal{T}$ are given by \eqref{eq:singulartorsion} and \eqref{eq:smoothtorsion} respectively.  The limiting smooth torsion is thus obtained as a limit of singular torsion currents supported on singular sets of measure zero. 

For example, given a smooth coframe  $\{\cof{1},\cof{2}\}$ for the unit square $\M=[0,1]^2$, we have by  Theorem \ref{thm:Homogenization} a sequence of closed singular frames $\{\cof{1}_{(k)},\cof{2}_{(k)}\}$ corresponding to an array of dislocations which converge to the co-frame $\{\cof{1},\cof{2}\}$ in the sense of currents. The corresponding torsion currents $\mathcal{T}_{(k)}$ act on functions by integration along the dislocation segments of the $k\times k$ dislocation array corresponding to $\cof{1}_{(k)}$.

\paragraph{Acknowledgments}
The authors would like to thank Reuven Segev and Cy Maor for many helpful discussions  and for revising our paper.

This research was partially funded by the Israel Science Foundation (Grant No. 1035/17), and by a grant from the Ministry of Science, Technology and Space, Israel and the Russian Foundation for Basic Research, the Russian Federation.

\bibliographystyle{spphys}
\bibliography{/Users/raz/Dropbox/tex/Refs/MyBibs}
%\bibliography{/Users/elikolami/Dropbox/Shared with Raz/Defects/bibtex/MyBibs-Elik}

%\input{referenc}
\end{document}